\documentclass[final,5p,times,twocolumn]{elsarticle}

\usepackage{enumitem}
\usepackage{xspace}
\usepackage{graphicx}
\usepackage{amsmath,amssymb,amsfonts,amsthm}
\usepackage{flushend}
\usepackage{bbold}
\usepackage{hyperref}

\newcommand{\R}{\ensuremath{\mathbb{R}}}
\newcommand{\N}{\ensuremath{\mathbb{N}}}

\newcommand{\bmat}[1]{\begin{bmatrix}#1\end{bmatrix}}
\newcommand{\smat}[1]{\left[\begin{smallmatrix}#1\end{smallmatrix}\right]}

\newtheorem{theorem}{Theorem}
\newtheorem{proposition}{Proposition}
\newtheorem{lem}{Lemma}

\newcommand{\newmarkedassumption}[1]{%
  \newenvironment{#1}
    {\pushQED{$\triangleleft$}\csname inner@#1\endcsname}
    {\popQED\csname endinner@#1\endcsname}%
  \newtheorem{inner@#1}%
}
\newmarkedassumption{assumption}{Assumption}
\newmarkedassumption{definition}{Definition}

\allowdisplaybreaks

\begin{document}

\begin{frontmatter}

\title{Event-triggered control of nonlinear systems from data}
\tnotetext[t1]{This publication is part
of the project Digital Twin with project number P18-03 of
the research programme TTW Perspective which is (partly)
financed by the Dutch Research Council (NWO). H.~Chen is supported by the China Scholarship Council.}
\tnotetext[]{Corresponding author: A.~Bisoffi.}

\author[Groningen]{Hailong Chen}\ead{hailong.chen@rug.nl}
\author[Groningen]{Claudio De Persis}\ead{c.de.persis@rug.nl}
\author[Milan]{Andrea Bisoffi}\ead{andrea.bisoffi@polimi.it}
\author[Firenze]{Pietro Tesi}\ead{pietro.tesi@unifi.it}

\address[Groningen]{ENTEG, University of Groningen, 9747 AG Groningen, The Netherlands}                                     
\address[Milan]{Department of Electronics, Information and Bioengineering, Politecnico di Milano, 20133 Milan, Italy}
\address[Firenze]{DINFO, University of Florence, 50139 Firenze, Italy}

\begin{abstract}
In a recent paper \cite{deposte2023}, we introduced a data-based approach to design event-triggered controllers for linear systems directly from data.
Here, we extend the results in \cite{deposte2023} to a class of nonlinear systems.
We provide two data-based designs certified by a (classical) Lyapunov function.
For these two designs, we devise event-triggered policies that rely on the previously found Lyapunov function, have parameters tuned from data, ensure a positive minimum inter-event time, and act based either on the state error or on the library error.
These two different policies, and their respective advantages, are illustrated numerically.
\end{abstract}


\end{frontmatter}

\section{Introduction}
Within control engineering, effectively managing nonlinear systems is crucial, especially in domains such as aerospace and robotics. Traditional control paradigms, typically time-triggered, can lead to inefficiency in both computation and resource utilization due to their predetermined operational instants. 
Conversely, event-triggered control initiates actions based on specific state conditions, potentially enhancing resource efficiency and system responsiveness.

Despite the advantages of event-triggered control, its application to nonlinear systems via conventional methods presents many challenges, mainly due to the complex dynamics that characterize nonlinear systems. 
Moreover, the design of feedback laws and event-triggering policies, whether through emulation-based approaches \cite{Tabuada07,Heemels2012} or co-design methods \cite{Postoyan2014,marchand2012general}, relies heavily on a dynamical model of the system. 
This dependence often limits the flexibility and applicability of event-triggered control when an accurate model is unavailable or difficult to derive. 
Recently, an alternative approach known as direct data-driven control has been developed to design controllers directly from input-state/output data bypassing the step of system identification. 
While this approach has primarily focused on linear systems \cite{de2019formulas, henk-ddctr-uncer, berberich2020robust-acc}, it has also been extended to nonlinear systems \cite{derotte2023,deptARC2023,martARC2023}. 
Inspired by such developments, there is a growing literature on applying these principles to data-driven event-triggered control see, e.g. \cite{deposte2023, iannelli2024hybrid,deng2024event}. 
By focusing on continuous-time linear time-invariant systems, \cite{deposte2023} derives data-driven versions of various popular event-triggered control policies in the literature \cite{Tabuada07,Heemels2012,Borgers-Heemels-tac14,Girard-tac15} and considers also disturbances acting on the system.
In~\cite{iannelli2024hybrid}, data-driven stabilization of linear time-varying systems is pursued by updating controller gain and Lyapunov function episodically, in an event-triggered fashion; between episodes, however, communication of plant and controller occurs continuously.
Other data-driven event-triggered methods have been developed for dynamic triggering mechanisms \cite{qi-et-al-tie2022} and predictive control \cite{deng2024event}. 
However, these studies predominantly focus on linear systems, whereas here we investigate data-driven event-triggered control strategies tailored for nonlinear systems.

Moreover, event-triggered control \cite{Tabuada07,Heemels2012,Postoyan2014,Liu-Jiang-aut15} often relies on the assumption of input-to-state stability (ISS) with respect to measurement errors. 
This assumption can limit the applicability of event-triggered control in practical scenarios where this stronger notion of stability cannot be assured. 
An ISS-based approach typically involves also knowing or constructing comparison functions that are closely linked to the system model, which are hard to rework when the system dynamics change, especially for nonlinear systems.
The co-design of control law and event-triggered law is achieved in \cite{xing2016event} without an ISS assumption, but the approach guarantees global boundedness of closed-loop signals and not asymptotic stability.  Also \cite{marchand2012general} does not consider ISS, but builds upon the existence of a control Lyapunov function.
On the other hand, our methods obtain quadratic Lyapunov functions from data.

\textbf{Considered framework.}
We aim to develop event-triggered control strategies for unknown input-affine nonlinear systems.
We assume availability of experimental input-state data generated offline by the unknown true system (\emph{ground truth}) and of a library of basis functions including those in the true system.
By adopting an emulation-based approach, we start with the design of a state feedback controller: the first design approach is based on nonlinearity cancellation, and the second one on contractivity. 
As detailed in \cite{derotte2023}, the first approach involves (approximate) nonlinearity cancellation, that is, the designed controller, derived from a data-based semidefinite program, aims at rendering the closed loop dominantly linear. 
With perfect nonlinearity cancellation, the  origin of the closed loop is globally asymptotically stable; otherwise locally asymptotically stable. 
As detailed in \cite{hu2024enforcing}, the second approach formulates data-based conditions to ensure contractive dynamics \cite{lohmiller1998contraction}, which lead to asymptotic stability of the origin.
We consider these two design approaches since they can be enhanced with a data-based design of the event-triggering policy, via their associated (classical) Lyapunov functions obtained from data.
The execution of the proposed feedback control laws is determined by one of two event-triggering rules, which are either based on an ``error-state threshold" or on an ``error-library threshold".
 
\textbf{Contribution.}  
We propose the design of an event-triggered control framework for a class of nonlinear systems that is overall data-based and, hence, does not need a dynamical model of the nonlinear system. 
The two design methods for the controller (nonlinearity cancellation and contractivity) are compared in the context of event-triggered control after we combine them with an event-triggering policy. Unlike most model- and data-based event-triggered schemes, we achieve 
stability properties of our scheme without resorting to an ISS assumption, thereby broadening the applicability of the proposed scheme.
At the same time, the existence of a minimum inter-event time is guaranteed by a data-based lower bound and Zeno behavior is excluded.

\textbf{Structure.} The formulation is introduced in Section \ref{sec:framework}. The data-based representation and the feedback controller design are presented in Section \ref{sec:controller}. Two different triggering mechanisms are derived in Section \ref{sec:policy}, where we establish asymptotic stability properties and the existence of a minimum inter-event time. Section \ref{sec:sim} provides numerical simulations.  

\textbf{Notation.}  $\mathbb N_0:=\{0,1,2,\dots\}$ is the set of non-negative integers and $\N:=\N_0\backslash\{0\}$. $\R$ is the set of real numbers.
For a matrix $M$, the notation
$M \succ 0$ ($M \succeq 0$) and $M \prec 0$ ($M \preceq 0$)
means that $M$ is symmetric positive and negative (semi-)definite, respectively.
The transpose of a matrix $M$ is $M^\top$. 
The identity and the zero matrices, whose dimensions depend on the context, are $I$ and $0$;
on the other hand, for $n \in \N$, $I_n$ is an $n$-by-$n$ identity matrix.
For a positive semidefinite matrix $A$, $A^{\frac{1}{2}}$ is its unique positive semidefinite square root.
We use $\|\cdot\|$ for the induced $2$-norm of a matrix.
For column vectors $x_1 \in \R^{n_1}$, \dots, $x_N \in \R^{n_N}$, $(x_1, \dots, x_N) := \smat{x_1^\top & \dots & x_N^\top}^\top \in \R^{n_1 + \dots + n_N}$.

\section{Problem formulation} \label{sec:framework}
\subsection{System and control objective}

We consider control-affine systems of the form
\begin{equation}
\label{system}
\dot x(t) = f(x(t)) + B u(t)
\end{equation}  
where: $t\in [0,\infty)$; $x(t) \in \mathbb R^{n}$ is the state and $u(t) \in \mathbb R^{m}$ is the control input for $n\in \N$, $m\in \N$; $f \colon \mathbb R^{n} \to \mathbb R^{n}$ is a smooth vector field with $f(0)=0$ and $B$ is a constant matrix. 
In the case of state-dependent input vector fields, i.e., $B$ replaced by some nonlinear map $x \mapsto g(x)$, a similar approach can be adopted by considering a dynamical controller of the type $\dot{u} = v$ and designing a control law for $v$ rather than for $u$, see \cite[\S V-B]{derotte2023}\footnote{In the present data-based event-triggered setting, this would require an actuator able to integrate the constant value of $v$ provided via the network until an updated value is provided, upon satisfaction of an event-triggering condition.}.

The goal is to design a map $\kappa$ and a sampling (triggering) policy so that  the control law
\begin{equation} \label{control}
u(t) = \kappa(x(t_k)), \quad t \in [t_k,t_{k+1}) 
\end{equation}
renders the origin of the closed-loop system asymptotically stable. 
Here, $\kappa$ represents a possibly nonlinear control map, while $\{t_k\}_{k \in \mathcal{N}}$ with $\mathcal{N}\subseteq\mathbb N_0$ is the sequence of 
sampling times, i.e., the sequence of times at which the control law is updated. 
The solutions to \eqref{system}-\eqref{control} are understood in the Carath\'eodory sense, so for any $k,k+1\in\mathcal{N}$ the solution flows on $[t_k,t_{k+1}]$ and experiences a jump at $t_{k+1}$. 
We will establish later in our main results in Section~\ref{sec:policy} that, for any $k,k+1\in\mathcal{N}$, $t_{k+1}-t_{k}$ is  lower bounded by a strictly positive constant independent of $k$ and of the initial condition so that the $t_k$'s do not accumulate (i.e., Zeno phenomenon does not occur).
Also, by solution to a dynamical system we always mean a maximal solution. 

\subsection{Scenario of interest}

Various solutions to this problem are available in the literature when $f$ and $B$ are known see, e.g., \cite{Tabuada07,Girard-tac15,Postoyan2014,Heemels2012,Abdelrahim2017,Dolk-et-al-tac17,dpt_DOS_2015,Wang-Lemmon-aut11,Liu-Jiang-aut15}. 
We are interested in the scenario where both $f$ and $B$ are unknown but we have access to experimental data and some priors on $f$. 
Firstly, we assume some prior information on $f$.

\begin{assumption} \label{ass:library} 
For $s \in \N$, we know a smooth function $\zeta\colon \mathbb R^{n} \rightarrow \mathbb R^{s}$ with $\zeta(0)=0$ such that, for some unknown matrix $A$, $f(x)=A\zeta(x)$ for all $x \in \R^n$.
\end{assumption}

Analogously to what is assumed in sparse system identification \cite{trajectory-matrix-application}, Assumption~\ref{ass:library} means that we know a library of functions that includes the \emph{ground truth}.
On the other hand, one can always consider more nonlinearities in $\zeta$ than those appearing in $f$, which would lead to more decision variables in the control design programs in Section~\ref{sec:controller}.

Recent work \cite{derotte2023} has considered this problem in a classical setting in which only the control law has to be designed. 
Here, the problem of designing also the sampling times, in addition to the nonlinear setting, makes the problem more challenging.
Under Assumption \ref{ass:library}, the dynamics can be written as
\begin{equation}
\label{systemZ}
\dot x(t) = A\zeta(x(t)) + B u(t)
\end{equation}  
with $A,B$ unknown. 

Assumption~\ref{ass:library} alone is not sufficient to derive any principled design methodology because it leaves $A$ and $B$ unconstrained. We therefore assume that information on the dynamics is obtained through experimental data. Specifically, we assume that we are given a sequence of data collected off-line 
\begin{equation*}
\mathbb D := \left\{ \big(u(t),x(t),\dot x(t) \big) \right\}_{t\in\{t_0, t_{1},\dots, t_{T-1}\}}
\end{equation*}
where $T \in \N$ is the number of samples. $\mathbb D$ consists of input, state and state-derivative data points collected from the system with one or multiple experiments. 
This means we have access to a set of input-state-state derivative samples verifying \eqref{system} for $t\in\{t_0, t_{1},\ldots, t_{T-1}\}$. 
The assumption that the data are clean is unrealistic but is made to keep the exposition simple. 
As a matter of fact, the results presented in this paper can be extended to the case of noisy data by following the same steps as in \cite[\S IV]{deposte2023} as we discuss in Section~\ref{sec:discussion}. Further, even the computation of the state derivative can be avoided \cite[Appendix A]{deposte2023} as we also discuss in Section~\ref{sec:discussion}.
Define the data matrices
\begin{align*}
\begin{array}{lllll}
U_0 & := \Big[ u(t_0) & u(t_{1}) & \cdots & u(t_{T-1})  \Big] , \\
X_0 & := \Big[ x(t_0) & x(t_{1}) & \cdots & x(t_{T-1})  \Big] ,\\
Z_0 & := \Big[ \zeta(x(t_0)) & \zeta(x(t_{1})) & \cdots & \zeta(x(t_{T-1}))  \Big],\\
X_1 & := \Big[ \dot x(t_0) & \dot x(t_{1}) & \cdots & \dot x(t_{T-1})  \Big]. 
\end{array}
\end{align*}
We assume the next condition on data richness \cite{willems2005,alsalti2023}. 

\begin{assumption} \label{ass:rich}
The matrix $\smat{U_0 \\ Z_0}$ has full row rank.
\end{assumption}

This assumption is verifiable from data; since $\smat{U_0 \\ Z_0} \in \mathbb R^{m+s \times T}$, the more data points $T$ one collects, the more columns $\smat{U_0 \\ Z_0}$ has and the higher the chance that $\smat{U_0 \\ Z_0}$ has full row rank.

\section{Learning a feedback controller from data}
\label{sec:controller}

To design the control system we proceed by \emph{emulation}~\cite{Heemels2012}: we first design a feedback controller that stabilizes the closed-loop system in the absence of sampling and we then design a triggering rule. 
For the design of the controller, we follow the approach in \cite{derotte2023} and derive a data-based representation of the closed-loop dynamics that is useful to determine suitable triggering rules later. 

\begin{lem} \label{lem:main}
Let Assumptions~\ref{ass:library}-\ref{ass:rich} hold. 
Consider any matrix $K \in \mathbb R^{m \times s}$. 
Then, \eqref{systemZ} with the control law $u=K\zeta(x)$ results in
the closed-loop dynamics 
\begin{equation*} \label{eq:GK_closed}
\dot{x}(t) = X_1 G \zeta(x(t))
\end{equation*}
where $G \in \mathbb R^{T \times s}$ is any solution to $\smat{K \\ I_s}=\smat{ U_0 \\ Z_0} G$.
\end{lem}
\begin{proof}
See \cite[Lemma 1]{derotte2023}.
\end{proof}

Lemma \ref{lem:main} gives a data-based representation of the closed-loop system when the controller takes the form $u=K\zeta(x)$. 
The reason for considering a control law of this form is to manage the nonlinearities in the system dynamics. 
We will demonstrate that this choice is effective and allows us to recreate two popular strategies for nonlinear controller design from data: Lyapunov’s indirect method (or linearization method) and contractivity. 
We discuss both strategies as they operate under different assumptions and may therefore be preferred depending on the specific problem at hand.

\subsection{Linearization method}
\label{sec:linearization}

Assume without loss of generality that $\zeta$ contains both linear and nonlinear functions and is partitioned as
\begin{equation} \label{eq:def_zeta}
\zeta(x) = \begin{bmatrix} x \\ Q(x) \end{bmatrix}
\end{equation} 
where $Q\colon \mathbb{R}^{n} \rightarrow \mathbb{R}^{s-n}$ contains all the nonlinearities. 
Further, as $\zeta$ is smooth and $\zeta(0)=0$ by Assumption~\ref{ass:library}, we  assume without loss of generality\footnote{
If this is not the case, a Taylor expansion of $Q$ around $0$ yields $Q(x) = \partial Q/\partial x(0) \,x  + r(x)$ with $r(x)$ a nonlinear map converging to zero faster than linearly; so, we can redefine $\zeta$ using $r$ in place of $Q$, see \cite{derotte2023}.
}
that $Q(x)$ is a higher-order term than $x$, i.e.,
\begin{equation}\label{prop-Q}
\lim_{x \rightarrow 0} \frac{|Q(x)|}{|x|} = 0. 
\end{equation}
By recalling Lemma \ref{lem:main} and
partitioning $G =: \begin{bmatrix} G_1 & G_2 \end{bmatrix}$, with
$G_1 \in \mathbb R^{T \times n}$ and $G_2 \in \mathbb R^{T \times (s-n)}$, the closed-loop dynamics is
\begin{equation} \label{eq:GK_closed_ext}
\dot{x}(t) = X_1G_1 x(t) + X_1 G_2 Q(x(t)). 
\end{equation}
Since $Q$ converges to zero faster than linearly, \eqref{eq:GK_closed_ext} provides a representation of the closed-loop dynamics that is suitable for control design. 
In particular, by Lyapunov's indirect method, a sufficient condition to stabilize the origin is that the matrix $X_1G_1$ is Hurwitz.
This can be translated into the program 
\begin{subequations}
\label{eq:NC}
\begin{align}
& \mathrm{minimize} && \| X_1 G_2 \|  \qquad (\text{for } P \succ 0,G,K)\label{eq:NCa}\\
& \mathrm{subject~to} && (X_1 G_1) P + P (X_1 G_1)^\top+ \Omega \preceq 0 \label{eq:NCb}\\
& && \begin{bmatrix} K \\ I_s \end{bmatrix} = \begin{bmatrix} U_0 \\ Z_0 \end{bmatrix} G \label{eq:NCc}
\end{align}
\end{subequations}
where $\Omega \succ 0$ is arbitrary.
In~\eqref{eq:NC}, \eqref{eq:NCb} imposes that $X_1 G_1$ is Hurwitz and \eqref{eq:NCc} that Lemma~\ref{lem:main} can be applied.
By allowing the controller $u = K \zeta(x) =: K_1 x + K_2 Q(x)$ to be nonlinear, its linear part $K_1 x$ stabilizes the linear part of the closed loop and its nonlinear part $ K_2 Q(x)$ minimizes the impact of the nonlinearities by way of \eqref{eq:NCa}, so that the closed loop is predominantly linear.
This is beneficial in enlarging an estimate of the basin of attraction in the sequel.
\eqref{eq:NC} can be made convex with standard manipulations by considering the decision variable $Y_1=G_1P$ in place of $G_1$, see \cite{derotte2023}.
We summarize the stability properties achieved by~\eqref{eq:NC}.
The proof is taken from \cite{derotte2023} and is repeated here for self-containedness.

\begin{proposition} \label{prop:NC}
Let Assumptions \ref{ass:library}-\ref{ass:rich} hold and
suppose that \eqref{eq:NC} is feasible and $K$ is returned. The origin of the closed-loop system
\eqref{system} with the control law $u=K\zeta(x)$ is asymptotically stable. 
\end{proposition}
\begin{proof}
By Assumptions \ref{ass:library}-\ref{ass:rich} and \eqref{eq:NCc}, \eqref{system} can be written as in \eqref{eq:GK_closed_ext}. Let $V(x):=x^\top S x$ for $S=P^{-1}$. By simple calculations,
\begin{align*}
& \frac{\partial V}{\partial x}(x) X_1 G \zeta(x) 
= 2 x^\top S \left( X_1 G_1 x + X_1 G_2 Q(x) \right) \\
&\leq - x^\top S \Omega S x + 2 x^\top SX_1G_2 Q(x) \quad \forall x \in \mathbb{R}^n
\end{align*}
by~\eqref{eq:NCb}.
The result holds by $S \Omega S \succ 0$ and $Q$ converging to zero faster than linearly.
\end{proof}

\subsection{Contractive design} \label{sec:contractive}

An alternative to linearization is to render the dynamics contractive \cite{lohmiller1998contraction,pavlov2004convergent}. We next recall such a property. 

\begin{definition}
System $\dot x = \mathcal{M} \zeta(x)$ is exponentially contractive on $\mathcal{X} \subseteq \mathbb{R}^n$ if there exist a matrix $S \succ 0$ and a scalar $\beta>0$ such that 
\begin{equation}
\label{contractivity}
\begin{array}{l}
\big(\mathcal{M} \displaystyle\frac{\partial \zeta}{\partial x}(x) \big)^\top S + S
\mathcal{M} \displaystyle\frac{\partial \zeta}{\partial x}(x) \preceq -\beta S
\quad \forall x\in \mathcal{X}. 
\end{array}
\end{equation}
In this case, we say that $\mathcal{X}$ is the contraction region with respect to the metric $S$.
\end{definition}

As with the linearization method, we can translate the property of $\dot x= X_1 G \zeta(x)$ being exponentially contractive on $\mathcal{X}$ into a data-based design program
\begin{subequations}
\label{eq:CONTR}
\begin{align}
& \mathrm{min.} && 0 \qquad (\text{for } P \succ 0,G,K) \label{eq:CONTRa}\\
& \mathrm{s.~t.} && 
\bmat{
(X_1 G_1) P + P (X_1 G_1)^\top + \Omega  & X_1 G_2 & P R_Q  \\ 
(X_1 G_2)^\top  & - I_{s-n} & 0 \\
(P R_Q)^\top & 0 & - I_{r}
} \preceq 0 \label{eq:CONTRb}\\
& && \begin{bmatrix} K \\ I_s \end{bmatrix} = \begin{bmatrix} U_0 \\ Z_0 \end{bmatrix} G 
\end{align}
\end{subequations}
where $\Omega \succ 0$ is arbitrary and $R_Q\in \mathbb R^{n \times r}$ is a known constant matrix such that 
\begin{equation} \label{asspt}
\frac{\partial Q}{\partial x}(x)^\top \frac{\partial Q}{\partial x}(x)\preceq R_Q R_Q^\top 
\quad \forall x\in \mathcal{X}. 
\end{equation} 
Since $Q(\cdot)$ is smooth, $R_Q$ exists when $\mathcal{X}$ is compact.
\eqref{eq:CONTR} can be made convex by standard manipulations.
We summarize the stability properties achieved via~\eqref{eq:CONTR} in the next result, after which we motivate the used contractive design. The proof is taken from \cite{hu2024enforcing} and is repeated here for self-containedness.
\begin{proposition} \label{prop:contractivity}
Let Assumptions \ref{ass:library}-\ref{ass:rich} hold and let $\mathcal{X}$ be a convex set containing the origin in its interior.
Suppose that \eqref{eq:CONTR} is feasible and $K$ is returned. Then, the origin of the closed-loop system
\eqref{system} with the control law $u=K\zeta(x)$ is asymptotically stable. 
\end{proposition} 
\begin{proof}
By Schur complement and $S:=P^{-1}$, \eqref{eq:CONTRb} gives 
\begin{equation}
\label{int-ineq}
S X_1 G_1 +(X_1 G_1)^\top S  + S \Omega S 
+  R_Q R_Q^\top  + S X_1 G_2 (X_1 G_2)^\top S \preceq 0.
\end{equation}
This proves that there exist a real $\varepsilon :=1$ and  matrices
\begin{align*}
& H := S X_1 G_1 +(X_1 G_1)^\top S  + S \Omega S, &&  J := S X_1 G_2,\\
& F := I_n, && O := R_Q R_Q^\top
\end{align*}
satisfying $H+\varepsilon J J^\top +\varepsilon^{-1} F^\top O F\preceq 0$. 
Now recall that under Assumptions \ref{ass:library}-\ref{ass:rich} the dynamics of the closed-loop system can be written as $\dot{x} = X_1 G \zeta(x)$. 
If the solution of \eqref{eq:CONTR} is such that $J=0$ then $(A+BK)\zeta(x)= X_1 G_1 x$ and, therefore, $(A+BK)\frac{\partial \zeta}{\partial x}(x)= X_1 G_1$. 
By \eqref{int-ineq}, the closed-loop dynamics is exponentially contractive with
$\beta :=\underline \lambda(S^{\frac{1}{2}} \Omega S^{\frac{1}{2}}) > 0$, 
where $\underline \lambda$ denotes the smallest eigenvalue.
If instead $J\ne 0$ then by the nonstrict Petersen's lemma \cite[Fact 2]{bisoffi2022data} we have that $H+\varepsilon JJ^\top +\varepsilon^{-1} F^\top O F\preceq 0$ implies $H+ J R^\top  F+ F^\top R J^\top\preceq 0$ for all $R$ such that $R R^\top\preceq O$. 
By combining this inequality with~\eqref{asspt} we obtain that for all $x\in \mathcal{X}$, $(X_1 G_1)^\top  S + S X_1 G_1 +  S \Omega S + 
S X_1 G_2 \frac{\partial Q}{\partial x}(x) + (X_1 G_2 \frac{\partial Q}{\partial x}(x))^\top S \preceq 0 $.
With $\beta=\underline \lambda(S^{\frac{1}{2}} \Omega S^{\frac{1}{2}})$, this implies $(X_1 G \frac{\partial \zeta}{\partial x}(x))^\top S + S X_1 G \frac{\partial \zeta}{\partial x}(x) \preceq -\beta S$ for all $x\in \mathcal{X}$, which implies again exponential contractivity.

We can now finalize the proof. Let $V(x) :=x^\top S x$. By the previous calculations we have
\begin{align}
& \frac{\partial V}{\partial x}(x) X_1G \zeta(x) 
= 2 x^\top S X_1G \zeta(x) \stackrel{(a)}{=} 2 x^\top S X_1G \big(\zeta(x)-\zeta(0) \big) \nonumber \\
&\stackrel{(b)}{=} 2 x^\top S X_1G \left[\frac{\partial \zeta}{\partial x}(\eta x)\right]_{\eta\in (0,1)} x \stackrel{(c)}{\le} -\beta x^\top S x \quad  \forall x\in \mathcal{X} \label{dotValongSolContr}
\end{align}  
where: $(a)$ holds by $\zeta(0)=0$; $(b)$ holds by the mean value theorem applied to the function $\eta \mapsto 2 x^\top S X_1 G \zeta(\eta x)$ from $[0,1]$ to $\mathbb{R}$; $(c)$ by $\mathcal{X}$ convex and the exponential contractivity property. Hence, any sublevel set $\mathcal{R}_\gamma : =\left\{x\in \mathbb{R}^n\colon x^\top S x \le \gamma \right\}$ of $V$, $\gamma \ge 0$, included in $\mathcal{X}$ is a forward invariant set for $\dot x = (A+BK) \zeta(x)$ and $\dot x = (A+BK) \zeta(x)$ is forward complete on $\mathcal{X}$.  Further, any solution of $\dot x = (A+BK) \zeta(x)$ initialized in $\mathcal{R}_\gamma$ exponentially converges to the origin \cite[Thm.~4.10]{khalil2002nonlinear}.
\end{proof}

Compared with the linearization method, contractive design entails stronger feasibility conditions since \eqref{eq:CONTRb} is more demanding than \eqref{eq:NCb}. 
Further, Lyapunov's indirect method allows us to treat the term $X_1 G_2 Q(x)$
that appears in the expression $\dot x = X_1 G_1 x + X_1 G_2 Q(x)$ as an \emph{unknown} quantity,
see \cite{derotte2023} for details. 
On the other hand, contractive design allows us to have \emph{global} stability properties if \eqref{contractivity} holds with $\mathcal{X}=\mathbb{R}^n$. 
This is instead not achievable (at least not by design) with Lyapunov's indirect method, and requires exact nonlinearity cancellation, {i.e.}, $X_1G_2=0$, which is possible only when the system to control has a specific structure \cite{derotte2023,gdpt2023cdc}.

Finally, we remark that \cite[Theorem 1]{hu2024enforcing}, from which Proposition~\ref{prop:contractivity} follows,  does not require an equilibrium at the origin for the unforced dynamics. 
On the contrary, we assume here that $\zeta(0)=0$. This implies that Proposition~\ref{prop:contractivity} is the data-based version of Krasovskii's stability theorem \cite{pavlov2004convergent}, \cite[Exercise 4.10]{khalil2002nonlinear}.%

\section{Learning a triggering policy from data} 
\label{sec:policy}

We discuss the digital implementation of the controller
\begin{equation} \label{controlDIG}
u(t) = K \zeta(x(t_k)), \quad t \in [t_k,t_{k+1}). 
\end{equation}
The closed-loop system of~\eqref{systemZ} and \eqref{controlDIG} is rewritten as
\begin{align}
\dot x(t) &=  A \zeta(x(t)) + B K \zeta(x(t_k))  = X_1 G \zeta(x(t)) +  B K e(t) \label{eq:closed-loop}
\end{align}
where  
\begin{equation} \label{eq:e}
e(t) := \zeta(x(t_k)) - \zeta(x(t)), \quad t \in [t_k,t_{k+1}),
\end{equation}
represents the sampling-induced error in computing $\zeta$ as a consequence of the mismatch between the last value of the state transmitted to the controller and its current value. 
The nonlinear setting in~\eqref{eq:closed-loop} extends the linear setting in~\cite{deposte2023}.
As customary in the control-by-emulation literature, $e$ is regarded
as a disturbance to the nominal dynamics, which is then ``controlled'' by a triggering condition so that stability is preserved despite infrequent sampling.

Let us make this argument precise.
Consider any of the two design methods described in Section~\ref{sec:controller}. By the previous argument, letting $V(x) :=x^\top S x$ with $S=P^{-1} \succ 0$ yields
\begin{align}
& \frac{\partial V}{\partial x}(x) (X_1 G \zeta(x) +  B K e) = 2 x^\top S \left( X_1 G \zeta(x) +  B K e \right) \nonumber \\
&  \leq - x^\top \Theta x + 2 x^\top \Phi Q(x) + 2 x^\top SBK e 
\quad  \forall x\in \mathcal{Y} \label{eq:closed-loopV3}
\end{align}
where
\begin{equation}
\label{quantities_for_2_designs}
\begin{aligned}
\mathcal{Y}  & =  \mathbb{R}^n,  &  \hspace*{-4pt}\Theta & = S\Omega S, & \hspace*{-4pt}\Phi & = S X_1 G_2 &  & \hspace*{-3pt}\text{for linearization method}; \\
\mathcal{Y} &=  \mathcal{X}, & \hspace*{-4pt}\Theta & =\beta S, & \hspace*{-4pt}\Phi & = 0 & & \hspace*{-3pt}\text{for contractive design};
\end{aligned}
\end{equation}
cf.~the proofs of Propositions \ref{prop:NC} and \ref{prop:contractivity}.
In the sequel, we will refer to these two cases for the quantities $\mathcal{Y}$, $\Theta$, $\Phi$.
In either cases, $\Theta \succ0$ and $x^\top \Phi Q(x)$ is a higher-order term than $x^\top \Theta x$. 
Accordingly, to ensure stability for the sampled-data implementation, it is sufficient that $|e| \leq \sigma |x|$ or $|e| \leq \sigma |\zeta(x)|$ for a sufficiently small constant $\sigma>0$. 
Both strategies have relative merits, so we will discuss both.

\subsection{Error-state threshold}
\label{sec:policy:error-state}

We start by considering a sampling strategy that guarantees $|e| \leq \sigma |x|$ for all times. 
A simple way to enforce this condition is via an event-triggering rule, originally developed in the model-based case \cite{Tabuada07} and which we now recreate in the context of data-driven control. To this end, we introduce the next parameterized matrix $\Psi(\cdot)$ and the vector $\nu$,
\begin{equation*}
\Psi(\sigma) := \begin{bmatrix} -\sigma^2 I_n & 0 \\ 
0 & I_s \end{bmatrix} \text{ and }
\nu := \begin{bmatrix} x \\ e \end{bmatrix}
\end{equation*}
with $\sigma > 0$ a design parameter to be determined.
Note that $\nu^\top \Psi(\sigma) \nu = - \sigma^2 |x|^2 + |e|^2.$
The sampling times are defined as
\begin{equation} \label{eq:triggering}
t_{k+1} := 
\begin{cases}
\inf \{t \in \mathbb R\colon t > t_k \text{ and } \nu(t)^\top \Psi(\sigma) \nu(t) = 0  \}  & \text{if } x(t_k) \neq 0 \\
+ \infty & \text{otherwise}
\end{cases}
\end{equation}
with $t_0$ given.
This logic ensures by design that $\nu^\top \Psi(\sigma) \nu \leq 0$  
along the solutions to \eqref{eq:closed-loop}, \eqref{eq:e}, \eqref{eq:triggering} as long as they
exist. 
We shall prove later that the sequence of 
sampling instants does not result in an accumulation point, which guarantees that a solution to \eqref{eq:closed-loop}, \eqref{eq:e}, \eqref{eq:triggering}  exists for all times (and is unique).
In Section~\ref{sec:discussion} we discuss how \eqref{eq:triggering} could be modified when the dynamics is perturbed by noise during closed-loop execution.

Closed-loop stability depends on $\sigma$, which must be chosen sufficiently small to  control the norm of error $e$. 
To this end, note that \eqref{eq:closed-loopV3} is equivalent to
\begin{align}
& \frac{\partial V}{\partial x}(x) (X_1 G \zeta(x) +  B K e) \notag  \\
& \le -\frac{1}{2} x^\top \Theta x + 2 x^\top \Phi Q(x) +
\begin{bmatrix} x \\ e  \end{bmatrix}^\top
\underbrace{\begin{bmatrix} -\frac{1}{2} \Theta  &  S X_1 L \\ (S X_1 L)^\top & 0 \end{bmatrix}}_{=:M}
\begin{bmatrix} x \\ e  \end{bmatrix} \label{eq:closed-loopV4}
\end{align}%
for all $x \in \mathcal{Y}$,
where $L$ is any solution to
\begin{equation} \label{eq:L}
\begin{bmatrix} K \\ 0 \end{bmatrix} = \begin{bmatrix} U_0 \\ Z_0 \end{bmatrix} L ,
\end{equation}
which exists under Assumption \ref{ass:rich} and implies $BK=X_1L$.
Recall that the event-triggering rule \eqref{eq:triggering} ensures that
$\nu^\top \Psi(\sigma) \nu \leq 0$ along the solutions to \eqref{eq:closed-loop}, \eqref{eq:e}, \eqref{eq:triggering}. 
Hence, asymptotic stability of the origin of the closed-loop system \eqref{eq:closed-loop}, \eqref{eq:e} is ensured if we select $\sigma$ to satisfy, for any $\nu\in\R^{n+s}$, the implication
\begin{equation} \label{eq:SDP_triggering_0}
\begin{array}{rlll}
\nu^\top \Psi(\sigma) \nu \leq 0 
& \Longrightarrow &   \nu^\top M \nu \leq 0 
 \end{array}
\end{equation} 
since, in~\eqref{eq:closed-loopV4}, $2 x^\top \Phi Q(x)$ is a higher-order term than $-\frac{1}{2} x^\top \Theta x$.
The next theorem provides a data-based condition to select $\sigma$ in \eqref{eq:triggering} to ensure \eqref{eq:SDP_triggering_0}, by~\eqref{eq:SDP_triggering}, and the desired stability property in turn.

\begin{theorem} \label{thm:exact_sampling}
Let Assumptions \ref{ass:library}-\ref{ass:rich} hold and consider the closed-loop system \eqref{eq:closed-loop}, \eqref{eq:e}, \eqref{eq:triggering}, where $K$ in \eqref{eq:closed-loop} is any solution to \eqref{eq:NC} or  \eqref{eq:CONTR}. 
Let $\mu >0$, $\sigma>0$ be a solution to the linear matrix inequality
\begin{equation}
\label{eq:SDP_triggering}
\mu M - \Psi(\sigma) \preceq 0
\end{equation}
in the decision variables $\mu > 0$ and $\sigma^2$, which is always feasible.
Then, the following holds.
\begin{itemize}[noitemsep,nosep,wide,topsep=-10pt]
\item[(a)] For any compact invariant set $\mathcal{R}$ there exists a \emph{minimum inter-event time}, i.e., for any solution $x$ with $x(0) \in \mathcal{R}$, the sequence of sampling instants 
satisfies $t_{k+1}-t_k \geq \tau(\sigma)$ for every $k \in \mathcal{N} \subseteq \mathbb N_0$, where $\tau(\sigma):=\frac{1}{\ell} \frac{\sigma}{1+\sigma}$ for some constant $\ell$ computed from data alone, as specified in the proof. 
\item[(b)] Let
\begin{align}
& \mathcal{V} := \{x\colon - \tfrac{1}{2} x^\top \Theta x + 2x^\top \Phi Q(x) < 0\}, \quad \mathcal{W} := \mathcal{V} \cap \mathcal{Y} \label{def_set_V_set_W}
\end{align}
with $\mathcal{Y}$ the set where \eqref{eq:closed-loopV3} holds.
Then, the origin of the closed-loop system is asymptotically stable and any set
\begin{align}
\label{def_set_Rgamma}
\mathcal{R}_\gamma := \{x\colon x^\top S x \leq \gamma \}
\end{align}
included in $\mathcal{W} \cup \{0\}$
is an invariant set and is an estimate of the basin of attraction.
\end{itemize}  
\end{theorem}
\begin{proof}
Let us show feasibility of \eqref{eq:SDP_triggering}. Take $\sigma=\sqrt{\mu/c}$
where $c>0$ is any sufficiently large constant such that $\Theta/2 - I_n/c \succ 0$ 
and $\Theta$ is $\Theta = S \Omega S \succ 0$, see \eqref{eq:NC} and below \eqref{eq:closed-loopV3}, or $\Theta = \beta S \succ 0$, see \eqref{eq:CONTR} and below \eqref{eq:closed-loopV3}. Such a $c$ exists because $\Theta \succ 0$.
For such a choice of $\sigma$, we have $-\mu \Theta/2 + \sigma^2 I_n = -\mu (\Theta/2-I_n/c) \prec 0$ 
for every $\mu>0$. Therefore, for this choice of $\sigma$, \eqref{eq:SDP_triggering}
is equivalent by a Schur complement to the two conditions
\begin{align*}
& -\mu (\Theta/2-I_n/c) \prec 0 \quad\text{and} \\
& -I_{s} + \mu (SX_1L)^\top (\Theta/2 - I_n/c)^{-1} (SX_1L) \preceq 0,
\end{align*}
which are jointly satisfied for $\mu$ sufficiently small. 

Let us prove item (a). Let $\mathcal{R}$ be any compact invariant set, and consider any
solution that starts in $\mathcal R$. 
The claim is trivial if $x(t_0)=0$. Conversely, pick any $x(t_k) \neq0$ where $k\in\mathbb{N}_0$. 
We show that if $x(t_k) \neq0$ then $x(t) \neq0$ for all $t \in [t_{k},t_{k+1})$. To see this, suppose by contradiction that $x(t_*)=0$ for some $t_* \in (t_{k},t_{k+1})$. Since $x$ is continuous on $(t_k,t_{k+1})$ and $\zeta$ is a smooth function satisfying 
\eqref{eq:def_zeta}-\eqref{prop-Q}, for each $\epsilon >0$ there exists $\delta \in (t_k,t_*)$ such that for each $t \in (\delta,t_*)$, $|x(t)|\leq \epsilon$ and $|\zeta(x(t))|\leq \epsilon$. Moreover, $|\zeta(x(t_k))| -|\zeta(x(t))| \leq |\zeta(x(t_k))-\zeta(x(t))| = |e(t)| \leq \sigma |x(t)|$ for all $t \in (\delta,t_*)$,
where the last inequality holds because $|e(t)| \leq \sigma |x(t)|$ as long as the solution exists.
Since $|x(t_k)| \leq |\zeta(x(t_k))|$, we thus have $
|x(t_k)| \leq |\zeta(x(t))| + \sigma |x(t)| \leq (1+\sigma) \epsilon$
for all $t \in (\delta,t_*)$; this leads to a contradiction as $\epsilon$ can be chosen arbitrarily small whereas $|x(t_k)| \neq 0$ is fixed.
This shows that if $x(t_k) \neq0$ then $x(t) \neq0$ for all $t \in [t_{k},t_{k+1})$, implying that $|e|/|x|$, used next in~\eqref{eq:Tab_MIET}, is well-defined over $[t_k,t_{k+1})$. As in \cite{Tabuada07}, we  bound inter-event times as
\begin{subequations}
\label{eq:Tab_MIET}
\begin{align}
& \frac{d}{dt} \frac{|e|}{|x|} = \frac{e^\top \dot{e}}{|e| |x|} - \frac{ |e| x^\top \dot{x}}{|x|^3} \leq  \frac{|\dot{\zeta}(x)|}{|x|} + \frac{ |e| |\dot{x}|}{|x|^2} \\
& \leq  \frac{\left\| \frac{\partial \zeta}{\partial x}(x) \right\| |\dot{x}|}{|x|} + \frac{ |e| |\dot{x}|}{|x|^2} 
\leq \ell_1 \left( 1+ \frac{|e|}{|x|} \right) \frac{|\dot{x}|}{|x|} \\
& =  \ell_1 \left( 1+ \frac{|e|}{|x|} \right) \frac{|(A+BK) \zeta(x) + BK e |}{|x|} \\
& \leq  \ell_2 \left( 1+ \frac{|e|}{|x|} \right) \frac{|\zeta(x)| + |e|}{|x|}  \leq  \ell \left( 1+ \frac{|e|}{|x|} \right)^2.
\end{align}
\end{subequations} 
where $\ell_1 := \max_{ x\in \mathcal{R}} \left\| \frac{\partial \zeta}{\partial x} (x)\right\|$ with $\ell_1\geq 1$ since $\left\| \frac{\partial \zeta}{\partial x} (x)\right\| \geq 1$ for all $x \in \mathbb{R}^n$ by~\eqref{eq:def_zeta}; $\ell_2 := \ell_1 \cdot \max \{ \|A+BK\|, \|BK\| \}$; $\ell := \ell_2 \cdot \sup_{x\in\mathcal{R} \backslash \{0\}} \frac{|\zeta(x)|}{|x|}$.
In particular, the last inequality follows because $\sup_{x\in \mathcal{R} \backslash \{0\}} \frac{|\zeta(x)|}{|x|} \geq 1$.
All these quantities exist and are finite since $\zeta$ is a smooth function, $\lim_{x \to 0} \frac{|\zeta(x)|}{|x|}=\lim_{x \to 0}  \sqrt{1+ \frac{|Q(x)|^2}{|x|^2}}
=\sqrt{1+ \lim_{x \to 0} \frac{|Q(x)|^2}{|x|^2}}=1$
and $\mathcal{R}$ is a compact set. Further, $\ell$ can be computed from data exploiting the fact that 
$A+BK=X_1G$ and $BK=X_1L$.

The rest of the proof is analogous to the model-based case, cf.~\cite{Tabuada07}.
Specifically, the time needed for $|e|/|x|$ to reach $\sigma$ is lower bounded by   
the time $\tau(\sigma)$ needed for $\phi$ to grow from $0$ (the value of $|e|/|x|$ right after each trigger) to $\sigma$, where $\phi$ is the solution to the differential equation $\dot \phi = \ell (1+\phi)^2$.
The expression of $\tau(\sigma)$ is given by $\tau(\sigma) := \frac{1}{\ell}  \frac{\sigma}{1+\sigma}$.
Thus, $t_{k+1}-t_k \geq \tau(\sigma)$. 
Since the solution $x$ and the instance $k$ are arbitrarily selected, $\tau(\sigma)$ is a minimum inter-event time valid over $\mathcal{R}$.

Let us prove item (b). As \eqref{eq:triggering} holds, 
$\nu^\top \Psi(\sigma) \nu \leq 0$ along the solutions to the closed-loop system. 
Combining this fact with \eqref{eq:closed-loopV4}, \eqref{eq:SDP_triggering_0} and \eqref{eq:SDP_triggering} we obtain
\begin{align}\label{eq:closed-loopV4a}
\frac{\partial V}{\partial x}(x) (X_1 G \zeta(x) + B K e) \le -\frac{1}{2} x^\top \Theta x + 2 x^\top \Phi Q(x)
\end{align}
for all $x \in \mathcal{Y}$, which is a nonempty set containing the origin since 
a stabilizing controller has been found by solving \eqref{eq:NC} or  \eqref{eq:CONTR}.
Now recall that $\Phi Q(x)$ is a higher-order term than $\Theta x$ for both the linearization method and contractive design (in the latter case, $\Phi=0$), and $\Theta \succ 0$.
Thus, $\mathcal{V}$ is nonempty; $\mathcal{W}:=\mathcal{V}\cap \mathcal{Y}$ is nonempty and $\mathcal{W} \cup \{0\}$ defines a neighbourhood of the origin. By standard Lyapunov arguments, any sublevel set 
$\mathcal{R}_\gamma$ included in $\mathcal{W} \cup \{0\}$
is an invariant set and an estimate of the basin of attraction. By item~(a), $\mathcal{R}_\gamma$ also has a guaranteed minimum inter-event time.
\end{proof}

Condition (\ref{eq:SDP_triggering}), which is linear in the  decision variables $\mu > 0$ and $\sigma^2>0$ and is always feasible, provides a data-based condition to design $\sigma$ so that the controlled system enjoys stability properties as well as a minimum inter-event time $\tau(\sigma)$. 
Condition (\ref{eq:SDP_triggering}) may be used to maximize $\sigma$ so to increase the inter-event times, see \cite[Proposition 3]{Postoyan_mystery} for further insights on the relation between the inter-event times and $\sigma$.

\subsection{Error-library threshold}
\label{sec:policy:error-library}

We now consider a sampling strategy ensuring $|e| \leq \sigma |\zeta(x)|$ for all times and discuss its advantages and disadvantages with respect to the error-state threshold. 
Similarly to before, let 
\begin{equation*}
\overline \Psi(\sigma) := \begin{bmatrix} -\sigma^2 I_s & 0 \\ 
0 & I_s \end{bmatrix} \text{ and }
\bar{\nu} := \begin{bmatrix} \zeta \\ e \end{bmatrix}
\end{equation*}
with $\sigma > 0$ a design parameter to be determined.
The sampling times are defined as
\begin{equation} \label{eq:triggering2}
t_{k+1} := 
\begin{cases}
\inf \{t \in \mathbb R\colon t > t_k \text{ and } \bar{\nu}(t)^\top \overline \Psi(\sigma) \bar{\nu}(t) = 0  \} & \text{if } x(t_k) \neq 0\\
+ \infty & \text{otherwise}
\end{cases}
\end{equation}
with $t_0$ given. 
As before, this logic ensures by design that $\bar{\nu}^\top \overline\Psi(\sigma) \bar{\nu} \leq 0$ along the solutions to~\eqref{eq:closed-loop}, \eqref{eq:e}, \eqref{eq:triggering2} as long as they exist, and we shall prove later that the sequence of sampling instants does not result in an accumulation point, which guarantees that a solution to \eqref{eq:closed-loop}, \eqref{eq:e}, \eqref{eq:triggering2} exists for all times (and is unique).
In Section~\ref{sec:discussion} we discuss how \eqref{eq:triggering2} could be modified when the dynamics is perturbed by noise during closed-loop execution.

Let $\eta > 0$ be any scalar selected by the designer.
Again, \eqref{eq:closed-loopV3} is equivalent to
\begin{align}\label{eq:closed-loopV2a}
& \frac{\partial V}{\partial x}(x) (X_1 G \zeta(x) +  B K e) \nonumber  \\
& \quad  \le -\frac{1}{2} x^\top \Theta x + \eta  Q(x)^\top Q(x) + 2 x^\top \Phi Q(x) \nonumber  \\
& \qquad + 
\left[
\begin{array}{c}
\zeta \\ \hline e
\end{array}
\right]^\top
\underbrace{
\left[ \begin{array}{cc|c} -\frac{1}{2} \Theta & 0 & S X_1 L \\ 
0 & -\eta I_{s-n} & 0 \\ \hline
(S X_1 L)^\top & 0 &0 \end{array} \right]
}_{=: \overline M}
\left[
\begin{array}{c}
\zeta \\ \hline e
\end{array}
\right]
\end{align}
for all $x \in \mathcal{Y}$, where $L$ is any solution to \eqref{eq:L}. 
Hence, asymptotic stability of the origin of the closed-loop system is ensured if we select $\sigma$ to satisfy, for all $\bar{\nu} \in\R^{2s}$, the implication
\begin{equation} \label{eq:SDP_triggering_0_bar}
\bar{\nu}^\top \overline \Psi(\sigma) \bar{\nu} \leq 0 \implies  \bar{\nu}^\top \overline M \bar{\nu} \le 0.
\end{equation}

\begin{theorem} \label{thm:exact_sampling2}
Let Assumptions \ref{ass:library}-\ref{ass:rich} hold and consider the closed-loop system \eqref{eq:closed-loop}, \eqref{eq:e}, \eqref{eq:triggering2}, where $K$ in \eqref{eq:closed-loop} is any solution to~\eqref{eq:NC} or \eqref{eq:CONTR}. 
Let $\mu >0$, $\sigma>0$ be a solution to the linear matrix inequality
\begin{equation}
\label{eq:SDP_triggering2}
\mu \overline M - \overline \Psi(\sigma) \preceq 0
\end{equation}
in the decision variables $\mu > 0$ and $\sigma^2$, which is always feasible.
Then, the following holds.
\begin{itemize}[noitemsep,nosep,wide,topsep=-10pt]
\item[(a)] For any compact invariant set $\mathcal{R}$ there exists a \emph{minimum inter-event time,} i.e., for any solution $x$ with $x(0) \in \mathcal{R}$, the sequence of sampling instants 
satisfies $t_{k+1}-t_k \geq \tau(\sigma)$ for every $k \in \mathcal{N} \subseteq \mathbb N_0$, where $\tau(\sigma):=\frac{1}{\omega} \frac{\sigma}{1+\sigma}$ for some constant $\omega$ computed from data alone, as specified in the proof. 
\item[(b)] Let
\begin{align}
& \!\mathcal{V} \!:= \!\{x\colon \! -\tfrac{1}{2} x^\top \Theta x + \eta Q(x)^\top Q(x) + 2x^\top \Phi Q(x) < 0\}, \notag \\
& \!\mathcal{W} := \mathcal{V} \cap \mathcal{Y}\label{def_set_V_set_W_2}
\end{align}
with $\mathcal{Y}$ the set where \eqref{eq:closed-loopV3} holds.
Then, the origin of the closed-loop system is asymptotically stable and any set%
\begin{align}
\label{def_set_Rgamma_2}
\mathcal{R}_\gamma := \{x\colon x^\top S x \leq \gamma \}
\end{align}
included in $\mathcal{W} \cup \{0\}$
is an invariant set and an estimate of the basin of attraction.
\end{itemize}
\end{theorem}
\begin{proof}
The proof follows the same steps as the proof of Theorem~\ref{thm:exact_sampling} but we report the details for completeness. 
Let us show the feasibility of \eqref{eq:SDP_triggering2}. 
Take $\sigma =\sqrt{\mu/c}$
where $c>0$ is any sufficiently large constant such that
\begin{equation*}
\bmat{\frac{1}{2} \Theta & 0 \\ 
0 &\eta I_{s-n}}  - \frac{1}{c} I_s
=: W - \frac{1}{c} I_s \succ 0
\end{equation*}
where $\Theta$ is $\Theta = S \Omega S \succ 0$, see \eqref{eq:NC} and below \eqref{eq:closed-loopV3}, or $\Theta = \beta S \succ 0$, see \eqref{eq:CONTR} and below \eqref{eq:closed-loopV3}.
Such a $c$ exists because $\Theta \succ 0$ and 
$\eta>0$.
For such a choice of $\sigma$, we have $-\mu W + \sigma^2 I_s = -\mu (W-I_s/c) \prec 0$ for every $\mu>0$. Therefore, for this choice of $\sigma$, \eqref{eq:SDP_triggering2}
is equivalent by a Schur complement to the two conditions
\begin{align*}
& -\mu (W-I_s/c) \prec 0 \quad \text{and} \\
& -I_s + \mu \begin{bmatrix} SX_1L \\ 0 \end{bmatrix}^\top (W - I_s/c)^{-1} \begin{bmatrix} SX_1L \\ 0 \end{bmatrix} \preceq 0,
\end{align*}
which are jointly satisfied for $\mu$ sufficiently small. 

Let us prove item (a). Let $\mathcal{R}$ be any compact invariant set, and consider any
solution that starts in $\mathcal R$. The claim is trivial if $x(t_0)=0$. Conversely, pick any $x(t_k) \neq0$ where $k\in\mathbb{N}_0$ and note that $x(t_k) \neq0$ implies $x(t) \neq0$ for all $t\in [t_k,t_{k+1})$ analogously as proved for Theorem~\ref{thm:exact_sampling}. In turn, this implies $\zeta(t) \!\neq\!0$ for all $t\in [t_k,t_{k+1})$ by~\eqref{eq:def_zeta}. 
As in \cite{Tabuada07}, we bound inter-event times as
\begin{align*}
& \frac{d}{dt} \frac{|e|}{|\zeta(x)|} =  \frac{e^\top \dot{e}}{|e| |\zeta(x)|} - \frac{ |e| \zeta(x)^\top \dot{\zeta}(x)}{|\zeta(x)|^3} \\
& \le \frac{|\dot{\zeta}(x)|}{|\zeta(x)|} + \frac{ |e| |\dot{\zeta}(x)|}{|\zeta(x)|^2} 
\le \ell_1 \left( 1+ \frac{|e|}{|\zeta(x)|} \right) \frac{|\dot{x}|}{|\zeta(x)|} \\
& = \ell_1 \left( 1+ \frac{|e|}{|\zeta(x)|} \right) \frac{|(A+BK) \zeta(x) + BK e |}{|\zeta(x)|} \\
& \le \omega \left( 1+ \frac{|e|}{|\zeta(x)|} \right) \frac{|\zeta(x)| + |e|}{|\zeta(x)|} = \omega \left( 1+ \frac{|e|}{|\zeta(x)|} \right)^2
\end{align*}
where $\ell_1 := \max_{ x\in \mathcal{R}} \big\| \frac{\partial \zeta}{\partial x}(x)\big\|$ with $\ell_1\geq 1$ since $\big\| \frac{\partial \zeta}{\partial x}(x) \big\| \geq 1$ for all $x\in\mathbb{R}^n$ by~\eqref{eq:def_zeta}; $\omega := \ell_1 \cdot \max \{ \|A+BK\|, \|BK\| \}$. As in the proof of Theorem \ref{thm:exact_sampling}, all these quantities exist and are finite since $\zeta$ is a smooth function and $\mathcal{R}$ is a compact set. Further, $\omega$ can be computed from data exploiting the fact that 
$A+BK=X_1G$ and that $BK=X_1L$.

Therefore, the time needed for $|e|/|\zeta(x)|$ to reach $\sigma$ is lower bounded by   
the time $\tau(\sigma)$ needed for $\phi$ to grow from $0$ (the value of $|e|/|\zeta(x)|$ right after each trigger) to $\sigma$, where $\phi$ is the solution to
the differential equation $\dot \phi = \omega (1+\phi)^2$.
The expression of $\tau(\sigma)$ is given by $\tau(\sigma) := \frac{1}{\omega}  \frac{\sigma}{1+\sigma}$.
Thus, $t_{k+1}-t_k \geq \tau(\sigma)$. Since the solution $x$ and the instance $k$ are arbitrarily selected, $\tau(\sigma)$ is a minimum inter-event time valid over $\mathcal{R}$. 

Let us prove item (b). By \eqref{eq:triggering2}, we have
$\bar{\nu}^\top \overline{\Psi}(\sigma) \bar{\nu} \leq 0$ along the solutions to the closed-loop system. Combining this fact with \eqref{eq:closed-loopV2a}, \eqref{eq:SDP_triggering_0_bar} and \eqref{eq:SDP_triggering2} we obtain
\begin{align}
&\frac{\partial V}{\partial x}(x) (X_1 G \zeta(x) +  B K e) \notag \\ 
&\quad \leq -\frac{1}{2} x^\top \Theta x + \eta Q(x)^\top Q(x) + 2 x^\top \Phi Q(x) \label{eq:closed-loopV4ab}
\end{align}
for all $x \in \mathcal{Y}$, which is a nonempty set containing the origin as we are assuming that a stabilizing controller has been found by solving \eqref{eq:NC} or  \eqref{eq:CONTR}. 
Now recall that $\Phi Q(x)$ is a higher-order term than $\Theta x$ and $\Theta \succ 0$. 
Thus, $\mathcal{V}$ is nonempty; $\mathcal{W}:=\mathcal{V}\cap \mathcal{Y}$ is nonempty and $\mathcal{W} \cup \{0\}$ defines a neighbourhood of the origin. By standard Lyapunov arguments, any sublevel set 
$\mathcal{R}_\gamma$ included in $\mathcal{W} \cup \{0\}$
is an invariant set and an estimate of the basin of attraction. 
By item~(a), $\mathcal{R}_\gamma$ also has a guaranteed minimum inter-event time.
\end{proof}

\subsection{Discussion}
\label{sec:discussion}

Let us discuss some points, among which the extensions of the results to the case of noisy measurements and neglected nonlinearities, which were not carried out on purpose to avoid distracting the reader from the main contribution of the article, namely, new data-based event-triggered schemes for nonlinear systems that do not require an ISS assumption. We envision that the analysis of these more complex scenarios can be carried out at the expense of simplicity.  

\subsubsection*{Noisy data and neglected nonlinearities}

Assuming noise-free data is unrealistic. However, dealing with noisy data appears to be feasible as  both the linearization method and the contractive design apply in the case of noisy data, see \cite[Theorem 6]{derotte2023} for the linearization method and \cite[Theorem 2]{hu2024enforcing} for the contractive design.
Similar considerations apply in the case of neglected nonlinearities, which the linearization method analyzes in \cite[Theorem 6]{derotte2023} and the contractive design in \cite[Theorem~3]{hu2024enforcing}.

\subsection*{Perturbed closed-loop dynamics}

When noise perturbs the dynamics not only during data collection but also during the closed-loop execution of the event-triggered control law, the setup becomes more challenging in terms of, e.g., guaranteeing the existence of a global minimum inter-event time with triggering rules like \eqref{eq:triggering} and \eqref{eq:triggering2}.
We believe that the proposed approach can be extended to the case of perturbed closed-loop dynamics by adopting the so-called mixed triggering condition considered in \cite[(27)]{deposte2023}. With the arguments in \cite[Thm.~4]{deposte2023}, such an extension would allow one to establish the existence of a robust positively invariant set for the closed-loop system. A detailed treatment of this extension is left for future work.
On the other hand, when the closed-loop dynamics is not perturbed, Theorems \ref{thm:exact_sampling} and \ref{thm:exact_sampling2} achieve the stronger property of regional asymptotic stabilization of the origin.

\subsubsection*{State derivatives}

To avoid the computation of state derivatives one could consider to take integral versions of the measurements as in \cite[Appendix A]{deposte2023}. Recently, novel data-driven derivative-free control techniques have been proposed in \cite{BOSSO2025101309,possieri2025derivative}, which may also find application in the context of this article.

\subsubsection*{Error-state threshold vs.~error-library threshold}

Compared with the error-state threshold, this strategy can achieve larger inter-event times. 
In fact, when considering the same $\mathcal{R}$, the constant $\omega$ of Theorem~\ref{thm:exact_sampling2} that controls the minimum inter-event time is related to the constant $\ell$ of Theorem~\ref{thm:exact_sampling} as $\ell=\omega \cdot \sup_{x\in\mathcal{R} \backslash \{0\}} |\zeta(x)|/|x|$ and, therefore, $\ell \geq \omega$ because $\zeta(x)$ also includes $x$ in the library, see~\eqref{eq:def_zeta}.
On the other hand, this strategy achieves worse control performance compared with the error-state threshold and this can be seen from the expression that describes the evolution of the Lyapunov function. 
In fact, with the second strategy we have an extra term $\eta Q(x)^\top Q(x)$ and this term negatively impacts the convergence rate and the size of the estimate of the basin of attraction.

\section{Numerical examples}
\label{sec:sim}

We numerically test our event-triggered designs, which yield asymptotic stability and an estimate of the basin of attraction in the first example (Section~\ref{sec:sim-poly}) and global asymptotic stability in the second one (Section~\ref{sec:sim-pend}).

\subsection{A polynomial system} \label{sec:sim-poly}

Consider the system from \cite[Example 14.11]{khalil2002nonlinear} given by
\begin{equation}\label{ex1}
\begin{bmatrix} \dot{x}_1\\\dot{x}_2
\end{bmatrix}=\begin{bmatrix}-x_1+x_1^2x_2\\0\end{bmatrix} + 
\begin{bmatrix} 0 \\ 1 \end{bmatrix} u,
\end{equation}
which has polynomial dynamics.
Suppose $\zeta(x)=(x_1,x_2,x_1^2,x_1^2x_2,x_1x_2^2,x_2^3)$, which is selected to emphasize that the library of functions can include functions that are \emph{not} present in the system dynamics. 
Accordingly, we write \eqref{ex1} as $\dot{x} = A\zeta(x)+Bu = \smat{-1&0&0&1&0&0\\0&0&0&0&0&0} \zeta(x)+ \smat{0\\1} u$.
We collect data by running an experiment of $1$ time unit with sampling period $0.1$, with initial conditions in the set $[-1,1]\times [-1,1]$ and the input generated by linear interpolation of samples of a random variable uniformly distributed in $[-20,20]$. 

For brevity, we consider the contractive design of Section~\ref{sec:contractive}, with the quantities in~\eqref{quantities_for_2_designs}.
Similar results can be obtained with the linearization method of Section~\ref{sec:linearization}. 
From the choice of $\zeta$, it follows $Q(x)=(x_1^2,x_1^2x_2,x_1x_2^2,x_2^3)$, thus
\begin{align*}
\frac{\partial Q}{\partial x}(x)^\top\frac{\partial Q}{\partial x}(x)& =
\bmat{4x_1^2+4 x_{1}^2 x_2^2+x_2^4 &2x_1^3x_2+2x_1x_2^3\\2x_1^3x_2+2x_1x_2^3&x_1^4+4x_1^2x_2^2+9 x_2^4} \\
& \preceq \bmat{0.2872&0\\0&0.0931}
\end{align*}
for all $x\in \mathcal{X}=\{x\in\mathbb{R}^2\colon |x_1|\leq\frac{1}{4},|x_2|\leq\frac{1}{4}\}$. 
Using $R_Q=\left[\begin{smallmatrix}
    0.2872&0.0000\\0.0000&0.0931
\end{smallmatrix}\right]^\frac{1}{2}$, see \eqref{asspt}, we solve \eqref{eq:CONTR}
and obtain $K=\smat{0.0008&-4.1514&0.0000 &-0.0013&0.0000  & 0.0000}$ and the Lyapunov function $V(x)=x^\top S x= x^\top \smat{0.3246 &0.0000\\0.0000&0.0327} x$.

\begin{figure}
\centerline{\includegraphics[width=0.85\linewidth]{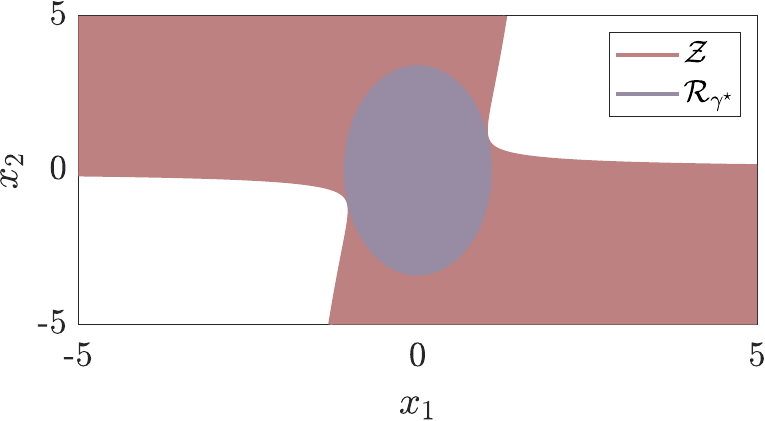}\vspace*{-3mm}}
\caption{The set $\mathcal{Z}$ and the largest estimate sub-level set $\mathcal{R}_{\gamma^\star}$ included in $\mathcal{Z}$.}\label{fig:error_state_set_Contraction}
\end{figure}

\subsubsection{The triggering policy with error-state threshold}

For the controller obtained with the contractive design, $\mathcal{V}$ in~\eqref{def_set_V_set_W} satisfies $\mathcal{V}=\left\{x\colon -\frac{1}{2}x^\top\Theta x<0\right\}=\R^2\setminus\{0\}$. 
In this case $\mathcal{X}\subseteq \mathcal{V}\cup \{0\}$, hence $\mathcal{W}$ in~\eqref{def_set_V_set_W} satisfies 
$\mathcal{W}\cup \{0\}=\mathcal{X}$ and any sublevel set $\mathcal{R}_\gamma$ of the Lyapunov function $V$ included in $\mathcal{X}$, see \eqref{def_set_Rgamma}, is an estimate of the basin of attraction (BOA), as certified by Theorem~\ref{thm:exact_sampling}. 
Actually, for $\mathcal{Z}:=\left\{x\colon\frac{\partial V}{\partial x}(x) X_1 G \zeta(x)  \le -x^\top \Theta x\right\}$, any sublevel set $\mathcal{R}_\gamma$ of the Lyapunov function $V$ included in $\mathcal{Z}$ is an estimate of the BOA\footnote{Recall that \eqref{eq:triggering} ensures $|e|\le \sigma |x|$ along solutions to the event-triggered closed loop.
From the proof of Theorem \ref{thm:exact_sampling}, see \eqref{eq:closed-loopV4a}, the derivative of $V$ along solutions to the event-triggered closed loop satisfies $\frac{\partial V}{\partial x} (x)(X_1 G \zeta(x) +  B K e) \le -\frac{1}{2}x^\top \Theta x$ for all $x\in \mathcal{X}$; hence, $\mathcal{R}_\gamma \subseteq \mathcal{X}$ is an estimate of the BOA. 
This follows from
\begin{align}
\label{set_incl_X_Z}
\frac{\partial V}{\partial x}(x) X_1 G \zeta(x) \le -x^\top \Theta x \quad \forall x\in \mathcal{X},
\end{align}
by~\eqref{dotValongSolContr}, and $-\frac{1}{2}x^\top \Theta x+2 x^\top S BK e\le 0$ for all $\smat{x\\e}$ with $|e|\le \sigma |x|$, by \eqref{eq:closed-loopV4}, \eqref{eq:SDP_triggering_0}, \eqref{eq:SDP_triggering}.
On the other hand, $\frac{\partial V}{\partial x}(x) X_1 G \zeta(x)  \le -x^\top \Theta x$ for all $x\in \mathcal{Z}$ and, thus, $\mathcal{R}_\gamma \subseteq \mathcal{Z}$ is also an estimate of the BOA.}.
Since $\mathcal{X}\subseteq \mathcal{Z}$ by~\eqref{set_incl_X_Z}, finding $\mathcal{R}_\gamma \subseteq \mathcal{Z}$ leads to a  larger estimate of the BOA than finding $\mathcal{R}_\gamma \subseteq \mathcal{X}$.
The largest $\mathcal{R}_\gamma \subseteq \mathcal{Z}$ is found numerically to correspond to $\gamma^\star= 0.365$ and, as just shown, is an estimate of the BOA.
The two sets $\mathcal{Z}$ and $\mathcal{R}_{\gamma^\star}$ are displayed in Fig.~\ref{fig:error_state_set_Contraction}.
We solve \eqref{eq:SDP_triggering}, 
which returns $\mu=0.3050$ and $\sigma=0.0402$. 
The resulting minimum inter-event time computed according to Theorem~\ref{thm:exact_sampling} is $\tau(\sigma)=1.2953\cdot 10^{-3}$. 
Fig.~\ref{fig:error_state_ETM_Contraction} reports the evolutions of $x(\cdot)$, $|e(\cdot)|$, $\sigma |x(\cdot)|$ and $\{t_{k+1}-t_k\}_{ k\in \mathcal{N}}$; the minimum inter-event time in \emph{this} simulation is $2.9829\cdot 10^{-3}$, which is larger but comparable with the guaranteed $\tau(\sigma)$.

\begin{figure*}
\centerline{
\includegraphics[width=0.31\linewidth]{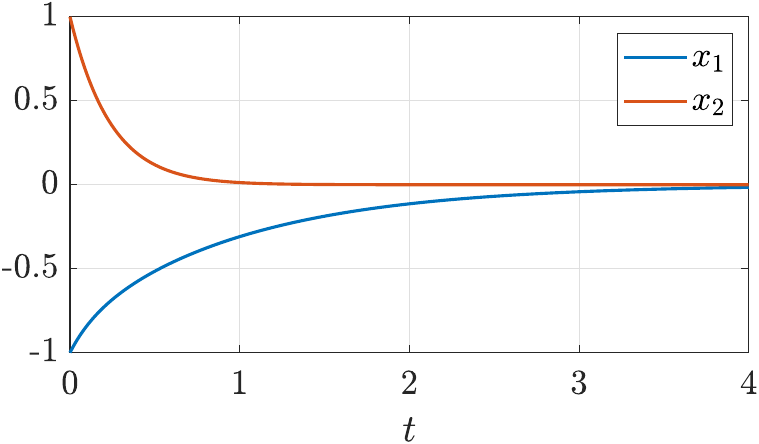}~
\includegraphics[width=0.31\linewidth]{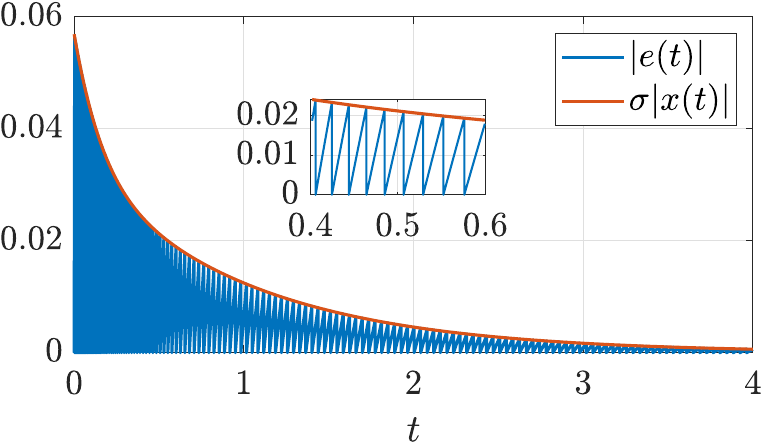}~
\includegraphics[width=0.31\linewidth]{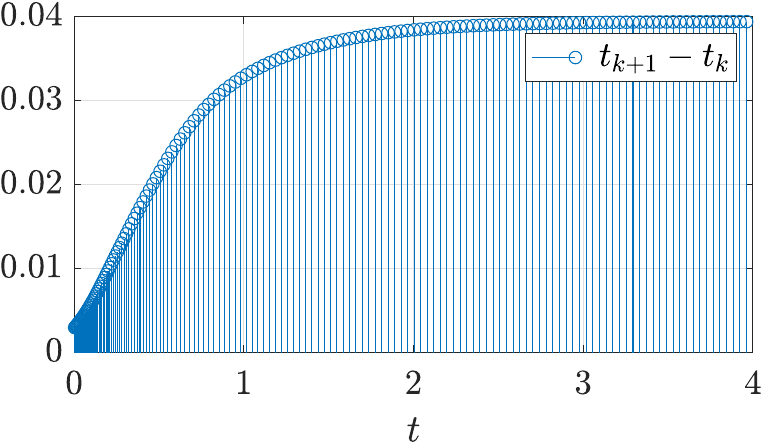}\vspace*{-3mm}
}
\caption{Left: state solutions. Middle: evolution of $|e(\cdot)|$ and $\sigma|x(\cdot)|$. Right: triggering inter-event times.}\label{fig:error_state_ETM_Contraction}
\end{figure*}

\begin{figure}
\centerline{\includegraphics[width=0.85\linewidth]{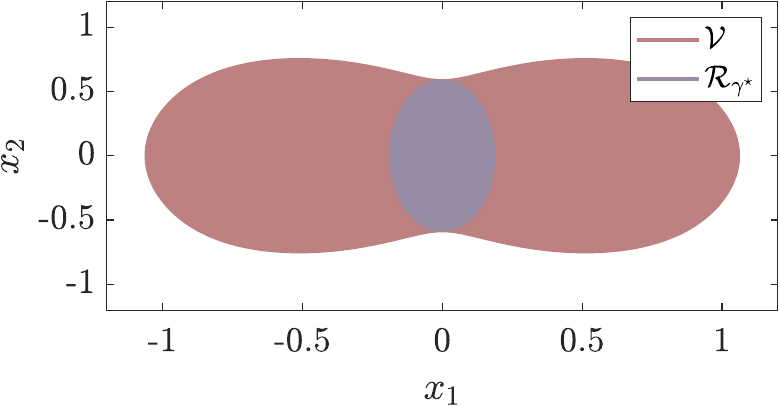}\vspace*{-3mm}}
\caption{
The set $\mathcal{V}$ and the largest estimate sub-level set $\mathcal{R}_{\gamma^\star}$ included in $\mathcal{V}$. }
\label{fig:error_library_set_contraction}
\end{figure}

\begin{figure*}
\centerline{
\includegraphics[width=0.30\linewidth]{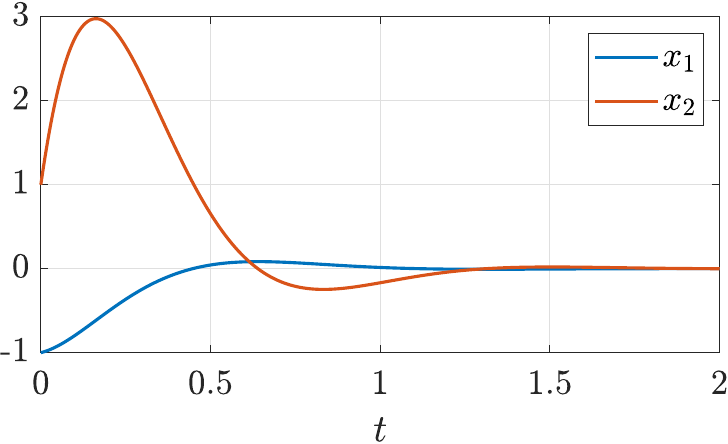}~
\includegraphics[width=0.32\linewidth]{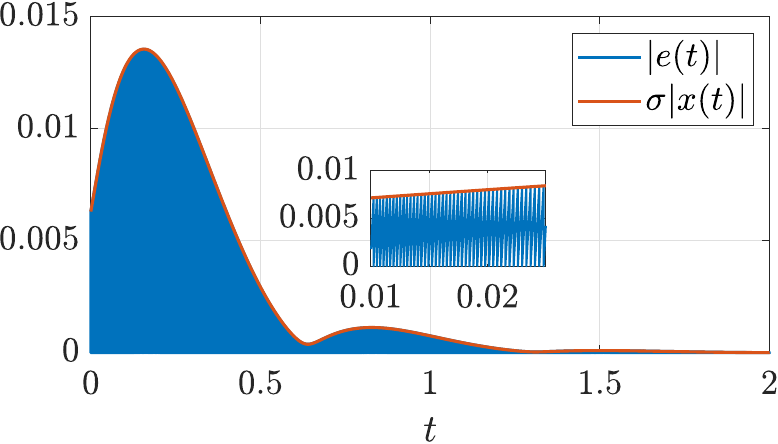}~
\includegraphics[width=0.30\linewidth]{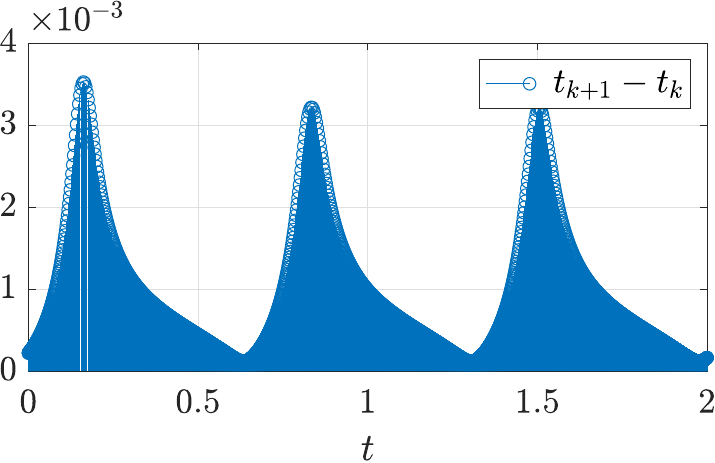}\vspace*{-3mm}
}
\caption{Left: state solutions. Middle: evolution of $|e(\cdot)|$ and $\sigma|x(\cdot)|$. Right: triggering inter-event times.}
\label{fig:inverted_pen_error_state_ETM_Contraction}
\end{figure*}

\subsubsection{The triggering policy with error-library threshold}

For the controller obtained with the contractive design, $\mathcal{V}$ in~\eqref{def_set_V_set_W_2} reduces to $\{x\colon-\frac{1}{2}x^\top\Theta x+\eta Q(x)^\top Q(x)<0\}$ and we select the design parameter $\eta=0.1$.
We check that $\mathcal{X}\subseteq \mathcal{V}\cup\{0\}$, hence 
$\mathcal{W}\cup \{0\}=\mathcal{X}$ from~\eqref{def_set_V_set_W_2} and any sublevel set $\mathcal{R}_\gamma$ of the Lyapunov function $V$ included in $\mathcal{X}$ is an estimate of the BOA, as certified by Theorem~\ref{thm:exact_sampling2}. 
Actually, any sublevel set $\mathcal{R}_\gamma$ of the Lyapunov function $V$ included in $\mathcal{V}\cup \{0\}$ is an estimate of the BOA by~\eqref{def_set_Rgamma_2}.
The largest $\mathcal{R}_\gamma \subseteq \mathcal{V}\cup \{0\}$ is found numerically to correspond to $\gamma^\star=0.0108$. 
The two sets $\mathcal{V}\cup \{0 \}$ and $\mathcal{R}_{\gamma^\star}$ are displayed in Fig.~\ref{fig:error_library_set_contraction}.
We solve \eqref{eq:SDP_triggering2}, which returns $\mu=0.3051$ and $\sigma=0.0402$.
The minimum inter-event time computed according to Theorem~\ref{thm:exact_sampling2} is $\tau(\sigma)=2.2435\cdot 10^{-3}$.
The evolutions of $x(\cdot)$, $|e(\cdot)|$, $\sigma |\zeta(x(\cdot))|$ and $\{t_{k+1}-t_k\}_{ k\in \mathcal{N}}$ are quite similar to those in Fig.~\ref{fig:error_state_ETM_Contraction}, and are thus not reported; the minimum inter-event time in this simulation is $5.0764\cdot 10^{-3}$. 
Overall, the policy with error-library threshold achieves, with respect to that with error-state threshold, a 1.7-fold increment of the minimum inter-event time but a smaller estimate of the BOA.
Despite this, the corresponding approach in Section~\ref{sec:policy:error-library} is relevant per se, beyond the specific numerical example.

\subsection{Inverted pendulum} \label{sec:sim-pend}

As second example, consider an inverted pendulum given by
\begin{equation}\label{ex3}
\dot{x}_1=x_2, \quad \dot{x}_2=\frac{g}{l}\sin x_1-\frac{h}{ml^2}x_2+\frac{1}{ml^2}u
\end{equation}
where $m$ is the mass of the pendulum, $l$ is the distance from the base to the center of mass of the pendulum, $h$ is the coefficient of rotational friction, and $g$ is the acceleration due to gravity. 
To generate the data points, we take $m=1$, $l=1$, $g=9.8$, and $h=0.01$. 
Suppose $\zeta(x)=(x_1,x_2,\sin x_1)$ so that \eqref{ex3} rewrites as $\dot{x}=\smat{0&1&0\\0&-0.01 &9.8} \zeta(x) + \smat{0\\1} u$.
We collect data by running an experiment of $1$ time unit with sampling period $0.1$, with initial conditions in the set $[-1,1] \times [-1,1]$ and the input generated by a linear interpolation of samples of a random variable uniformly distributed in $[-1,1]$. 
As in Section \ref{sec:sim-poly}, we only consider the contraction design of Section~\ref{sec:contractive}. From the choice of $\zeta$, it follows $Q(x)=\sin x_1$, thus $\frac{\partial Q}{\partial x}(x)^\top \frac{\partial Q}{\partial x}(x)
=
\smat{(\cos x_1)^2 &0\\0&0}
\preceq
\smat{
1 &0\\0&0
}$ for all $x\in\mathcal{X}=\mathbb{R}^2$.
Using $R_Q=\smat{1&0\\0&0}$, we solve \eqref{eq:CONTR} and obtain 
$K=\left[\begin{smallmatrix}-35.7625&  -7.4219&-9.6214\end{smallmatrix}\right]$.
For this example, the origin of the closed-loop system \eqref{ex3} with $u=K \zeta(x)$ is globally exponentially stable, by Proposition~\ref{prop:contractivity}.

\subsubsection{The triggering policy with error-state threshold}
We solve \eqref{eq:SDP_triggering}, which returns $\mu=0.0037$ and $\sigma=0.0045$.
The resulting minimum inter-event time is $\tau(\sigma)=6.8352\cdot 10^{-5}$.
Fig.~\ref{fig:inverted_pen_error_state_ETM_Contraction} reports the evolutions of $x(\cdot)$, $|e(\cdot)|$, $\sigma |x(\cdot)|$ and $\{t_{k+1}-t_k\}_{ k\in \mathcal{N}}$. For this example, Theorem \ref{thm:exact_sampling} guarantees that the origin is globally asymptotically stable for the system \eqref{ex3} when the controller is digitally implemented, i.e., when $u(\cdot)$ is given by \eqref{controlDIG} and \eqref{eq:triggering}. 
Moreover, the minimum inter-event time is globally defined, i.e., for any solution $x(\cdot)$, the sequence of sampling instants satisfies $t_{k+1}-t_k\ge \tau(\sigma)$ for all $k\in \mathcal{N}$.

\subsubsection{The triggering policy with error-library threshold}

For this controller and $\eta=0.1$, we solve \eqref{eq:SDP_triggering2}, which returns 
$\mu=0.0037$ and $\sigma=0.0045$. The resulting minimum inter-event time is
$\tau(\sigma)=8.3868\cdot 10^{-5}$.
The evolutions of $x(\cdot)$, $|e(\cdot)|$, $\sigma |\zeta(x(\cdot))|$ and $\{t_{k+1}-t_k\}_{ k\in \mathcal{N}}$ are quite similar to those in Fig.~\ref{fig:inverted_pen_error_state_ETM_Contraction}, and are thus not reported; the minimum inter-event time in this simulation is $1.6952\cdot 10^{-4}$.
For this example, Proposition~\ref{prop:contractivity} returns $\Theta=\left[\begin{smallmatrix}0.8751&0.1445\\0.1445&0.0427
\end{smallmatrix}\right]$. 
Hence, by \eqref{eq:closed-loopV4ab}, $\frac{\partial V}{\partial x}(x) (X_1 G \zeta(x) +  B K e) \le -\frac{1}{2} x^\top \Theta x + \eta Q(x)^\top Q(x)=-\frac{1}{2}x^\top\Theta x  +0.1(\sin x_1) ^2\leq-0.3376x_1^2-0.0213x_2^2-0.1445 x_1x_2< 0$ holds for all $x\in\mathbb{R}^2\setminus\{0\}$. 
Also with this triggering policy, the minimum inter-event time is globally defined in this example.

\section{Concluding remarks}
\label{sec:conc}

We have presented an approach to design event-triggered control schemes for nonlinear systems directly from data. The approach follows a so-called emulation scheme, where the control law is first designed without considering sampling, and then the communication policy is designed. 
This scheme is flexible as it allows the control law to be designed according to any design strategy; we have considered a linearization method and a contractive design. 
We also proposed two different communication policies and discussed their respective properties.
None of these policies relied on an ISS assumption.

Designing event-triggered control schemes for nonlinear systems directly from data is a complex problem. We believe that the results presented in this work can provide a foundation for addressing even more complex problems, such as the case of output feedback control, and, more generally, for developing methods that automatically learn to control and communicate over a network.

\bibliographystyle{elsarticle-harv} 

\bibliography{refs}

\end{document}